\documentclass[10pt,twocolumn,balancelastpage,twoside]{IEEEtran}
\usepackage{mathtools}
\usepackage{amsmath,amssymb,amsthm,mathrsfs}
\usepackage{dsfont}
\usepackage{psfrag}
\usepackage{nccmath}
\usepackage{bm}
\usepackage{color}
\usepackage{cite}
\usepackage{stfloats}
\usepackage{float}
\usepackage{lipsum}
\usepackage{multicol}
\usepackage{hyperref}
\usepackage{pst-poly,pst-node,pst-plot,pst-grad}
\usepackage{multirow}
\usepackage{cite,url}
\usepackage{enumerate}
\hypersetup{bookmarks=false,unicode=false,colorlinks=false}
\usepackage[enable-write18, interaction=nonstopmode]{auto-pst-pdf}
\usepackage{ifplatform, pst-pdf, xkeyval}

\usepackage{lipsum}

\theoremstyle{definition}
\newtheorem{definition}{Definition}
\theoremstyle{plain}
\newtheorem{theorem}{Theorem}
\newtheorem{lem}{Lemma}

\newtheorem{prop}{Proposition}

\theoremstyle{remark}
\newtheorem{remark}{Remark}

\title{Power Management Policies for AWGN Channels with Slow-Varying Harvested Energy}
\author{Ali Zibaeenejad, \emph{IEEE Member}\\
\emph{Email}: azibaeen@ieee.org}
\begin{document}
\maketitle \thispagestyle{empty} \pagestyle{empty}

\begin{abstract}
In this paper, we study power management (PM) policies for an Energy Harvesting Additive White Gaussian Noise (EH-AWGN) channel. The arrival rate of the harvested energy is assumed to remain unchanged during each data frame (block code) and to change independently across block codes. The harvested energy sequence is known causally (online) at the transmitter.
The transmitter is equipped with a rechargeable battery with infinite energy storage capacity. The transmitter is able to adapt the allocated energy and the corresponding transmission rate of each block according to a PM policy.
Three novel online PM policies are established. The policies are universal, in the sense of the distribution of the harvested energy, and simple, in the sense of complexity, and asymptotically optimal, in the sense of maximum achievable average rates (throughput) taken over a long-term horizon of blocks.
\end{abstract}

\section{Introduction}
By 2025, Internet of Things (IoT) connected devices installed base worldwide will exceed $75$ billion devices, such as wireless sensors, tablets and smart-phones, worldwide as well as data rate up-to $10~Gb/s$ per person~\cite{statistica}. Affording continuous long-lasting energy for such devices with high data rates is a real challenge of the future Internet of Things (IoT) network. Supplying the energy of the IoT devices from green resources, such as wind, solar, and from traveling wireless signals, such as interference signals and television signals, has recently received extensive attention. Green communication recommends employing the IoT devices with energy harvesting (EH) capabilities and rechargeable batteries, because they reduce conventional
fossil energy usage which eventually produces less carbon dioxide, and they resolve the following communication concerns: energy self-sufficiency, energy self-sustainability, and ability to deploy in places with no electric power grids or at power outage occasions. Despite these benefits, the varying nature of the absorbed energy and lack of information about the status of the energy arrival in future make the design of a EH communication system a serious challenge. A power management (PM) policy is required to decide what portion of the absorbed energy is to be assigned to the current data frame and what portion of the absorbed energy is to be stored
in the battery for the future use when energy shortage is likely.
In this technology, the transmitter is able to adapt its communication data rate of each frame, according to the assigned power by the policy, to assure that a reliable communication takes place. A standard performance benchmark for a policy is the maximum \emph{average data rates} (throughput) achieved by that policy, where the average is taken over $L$ block rates.\\
\indent In seminal paper~\cite{ulukus}, Ozel and Ulukus studied
the fundamental limits of a point-to-point Energy Harvesting Additive White Gaussian Noise (EH-AWGN)  channel from a transmitter, which is equipped with an infinite size rechargeable battery, to a receiver. Two essentially different perspectives of one model is studied in this work: First, the Shannon capacity~\cite{cover} problem of the EH-AWGN; second, the THroughput Maximization (THM)~\cite[(13)]{ulukus} problem. The first problem looks for a fixed maximum achievable rate at which reliable communication is guaranteed for any block codes by using a single code-book. On the other hand, the second problem allows using an individual code-book based on an exclusive achievable rate for each block code. A PM policy manages the total available energy, including the stored energy and harvested energy, such that the average of these achievable rates (throughput) is maximized. The THM problem is useful and interesting for real situations where the energy arrival rates are (approximately) constant across symbols but they vary across blocks. The reality of this model is due to slow-varying nature of the energy resources.\\
\indent In a standard AWGN channel with no EH capability and the same power constraint on each block, both problems trivially lead to the same rate. However, the existence of the battery in an EH-AWGN channel emerges dependency between the power constraints of the blocks: spending or saving energy during a data frame impacts the available energy for its \emph{future} blocks, and thus the rates of the blocks are dependent. Also, this model is different from a standard parallel AWGN channel with $L$ paths, because the energy causality (EC) constraint~\cite[(1)]{ulukus} on the EH-AWGN channel, which states the energy can not be employed before it is harvested, makes a clear distinction: Not only the total harvested energy during each time frame is important, but also the order of the harvested energy sequence matters. This constraint makes the THM problem even more challenging when the entire arrival energy sequence is not known in advance (online case).\\
\indent The Shannon capacity problem and THM problem of an EH-AWGN channel have been extensively studied in the literature (See~\cite{ulukus, yener, zhang12, zhang14, review15, ozgur15} and the references therein). We briefly review the most related points as follows.\\
\indent The Shannon capacity of an energy harvesting AWGN channel with an infinite-size battery has been established by Ozel and Ulukus~\cite{ulukus}. They showed that the capacity of the EH-AWGN channel with average harvested energy rate $\overline{\mathcal{E}}$ is the same as that of the classical AWGN with an average power constraint equal to $\overline{\mathcal{E}}$. They have developed two remarkable coding schemes for the \emph{capacity problem}: Save-And-Transmit (SAT) and Best-Effort-Transmit (BET). These schemes manage the power allocation across symbols of block codes along with a Gaussian code-book for data transmission. Also, they studied the THM problem for a non-causal model in which the realization of the entire energy arrival sequence is known in advance. They have developed an offline Optimal PM (OPM) policy across block codes, which was originally given in~\cite{modiano} for the context of energy minimal transmission in a delay-limited scenario. Ozel et.~al.~\cite{yener} have extended this work to optimal policies for the THM problem of a fading channel with causally known channel gains. They have designed a novel offline policy based on the directional water-filling (DWF) approach (A similar approach under the name ``staircase water-filling algorithm'' have been developed in~\cite{zhang12}). These optimal offline policies~\cite{modiano, ulukus, yener, zhang12} keeps power transmission as constant as possible across blocks. The computational complexity of these policies grows (at least) linearly with increment of $L$.\\
\indent The THM problem is called online if the realization of the energy arrival sequence is only known up to the current time but not more. Authors of papers~\cite{modiano, yener, sinha12, wang13} have studied online policies in which they have modeled the recharge rate by Markov Decision Process (MDP) and solved the problem numerically by using Dynamic Programming (DP) technique without enough engineering insights for the policy structure. On the other hand, the complexity of solutions increases as $L$ grows such that they become practically infeasible (For a detailed critique on the approach of these papers, see~\cite{ozgur16}). Publications~\cite{ozgur16, ozgur17-conf, Ulukus2017-ISIT} have studied the online THM problem when the transmitter is provided with a finite size battery. They have acquired simple online policies based on the average rate $\overline{\mathcal{E}}$. However, the offered policies have generally a constant gap with the Upper Bound (UB) \emph{independent of the problem parameters}, i.e., even if the battery storage capacity is infinite.\\
\indent In this paper, we consider a point-to-point EH-AWGN channel with an infinite size battery. Practically, if the energy storage capacity of the battery is relatively much larger than the average of energy arrival distribution, the battery size can be considered to be infinite. We assume that the harvested energy rate is constant during each block, and it changes across blocks according to an independent identically distributed (i.i.d.) sequence with some \emph{arbitrary} known distribution (similar to \cite{ozgur16}). This sequence is causally known at the transmitter. We study the online THM problem and the corresponding PM policies of this model. This work is an extension of our recent paper \cite{zib-IWCIT} in which the distribution of the energy arrival sequence is Bernoulli (similar to \cite{ozgur17-conf}). The results of this paper hold for any arbitrary distribution of energy arrival sequence and thus the proposed policies are universal.\\
\indent We establish three online PM policies in this work with the following properties:
1) They all are optimal in the asymptotic sense $L\rightarrow \infty$. Hence, the derived policies of this work outperforms sub-optimal policies \cite{ozgur16, ozgur17-conf, Ulukus2017-ISIT} for case infinite battery size; 2) Their order of complexity is constant $(O(1))$ as $L$ grows; 3) They all meet the offline OPM policy~\cite{modiano, ulukus, yener}, in the asymptotic sense $L\rightarrow \infty$. Hence, the proposed online policies can be employed as optimal offline policies with complexity $O(1)$ in this asymptotic sense; 4) The proposed policies can be exploited as simple offline policies with close performance to the offline OPM~\cite{modiano, ulukus, yener} at typical finite values of $L$, as it is illustrated by simulations; 5) They can be universally utilized for any energy arrival distribution because the knowledge of the average of the energy arrival distribution ($\overline{\mathcal{E}}$) at the transmitter is sufficient for the policies; 6) The structure of the policies is fundamentally different from previously known PM policies in the literature.\\
\indent The organization of this paper is as follows: In Section~\ref{Sec:Model}, we state the problem definitions and the studied model. In Section~\ref{Sec:Results}, we establish the main results of this paper. In Section~\ref{Sec:Numerical-Results}, we present the numerical results to compare our innovative methods with major known results. In Section~\ref{Sec:Conclusion}, we finally conclude this paper.

\section{System Model} \label{Sec:Model}
Assume a point-to-point EH-AWGN channel. The transmitter (TX) affords the energy of the transmission by exogenous energy arrivals harvested from the environment. The TX is supplied with a battery with an infinite size, which enables the TX to store the harvested energy. The transmission consists of $L$ block data frames (block codes) such that each block contains $n$ symbols, where $L\gg1$ and $n$ is sufficiently large to assure that information-theoretic coding rate is achievable.\\
\indent Suppose that the harvested energy arrival rate remains constant during each block code transmission and it changes i.i.d. across block codes. The energy arrival rate (absorbed power) in block code $\ell$ in denoted by $E_\ell$ (Watts), where $1\leq \ell \leq L$, and $\{E_k\}_{k = 1}^L$ is a sequence drawn  i.i.d. according to distribution $P_E$ \emph{across} blocks. For simplicity, we assume that the duration of each \emph{symbol} is one unit time. Hence, the harvested energy during block code $\ell$ is $nE_{\ell}$ (Joules).

\begin{definition}
The EH model is called offline, if the TX knows the realization of $\{E_k\}_{k = 1}^L$ non-causally at the \emph{beginning} of the transmission. The EH problem is called online, if the TX knows the realization causally: $\{E_k\}_{k = 1}^\ell$ is available to the TX at the beginning of block code $\ell^{th}$.
\end{definition}

Let random variable $X_{\ell j}$ represent the transmission symbol $j^{th}$ in block code $\ell^{th}$, where $j \in \{1, \ldots, n\}$ and $\ell \in \{1, \ldots, L\}$.
The Energy Causality (EC) constraint is
\begin{equation}
  \sum_{k=1}^{\ell-1} \sum_{i=1}^{n} X_{k i}^2 + \sum_{i=1}^{j} X_{\ell i}^2 \leq \sum_{k=1}^{\ell-1} n E_{k} + j E_{\ell} \: \label{power-const}
\end{equation}
That is, the sent energy at each time instant does not exceed the total available energy till that time instant.
Denote the transmission power of block code $\ell$ by $Q_\ell$, where
\begin{equation}
  Q_\ell = \frac{1}{n}\sum_{i=1}^n X_{\ell i}^2 \: .
\end{equation}
Also, let $B_\ell$ be the energy stored in the battery at the beginning of block code $\ell$. Assuming initial charge $B_1 = 0$, the sequence of the battery charge is given by
\begin{equation}
  B_{\ell+1} 
             = B_\ell + n (E_\ell -  Q_\ell)\: .
\end{equation}
As depicted in Fig.~\ref{Model-pic}, if $X_{\ell i}$ is sent, the receiver detects $Y_{\ell i} = X_{\ell i} + N_{\ell i}$, where $N_{\ell i}$ is Gaussian noise with zero mean and (normalized) variance $\sigma^2 = 1$. The noise is distributed i.i.d. across symbols with the same Gaussian distribution. The transmitter is allowed to apply any power allocation across symbols or block codes as long as \eqref{power-const} is met. Also, the transmitter is permitted to code each block based on an individual code-book according to the following definition:
\begin{figure}[t]
  \centering
  \includegraphics[width=9cm]{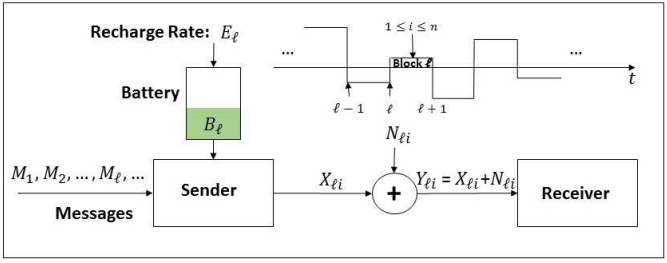}
  \caption{The AWGN energy harvesting model with slow varying energy arrivals. Sequence $\{E_\ell\}_{\ell = 1}^L$ is known causally at the transmitter.}\label{Model-pic}
\end{figure}
\begin{definition}
Let $R_\ell$ be the information rate in block code $\ell^{th}$. That is, message $M_\ell \in \{1, \ldots, 2^{nR_\ell}\}$ is to be sent in this block, where $\{M_\ell\}_{\ell = 1}^{L}$ is an i.i.d. sequence with a uniform distribution.
An admissible coding scheme for block code $\ell$ consists of an encoder, a decoder, and a code-book $\mathcal{C}_\ell^{(n)}$. The TX selects encoder $\mathcal{F}^{(\ell)}$ to sent block of symbols $X_{\ell 1}^n = \{X_{\ell i}\}_{i=1}^n$ by $X_{\ell 1}^n = \mathcal{F}^{(\ell)}(M_\ell, B_{\ell - 1}, E_{\ell})$ subject to~\eqref{power-const}. The decoder selects decoder $\mathcal{D}^{(\ell)}$ to decode the received sequence $Y_{\ell 1}^n = \{Y_{\ell i}\}_{i=1}^n$ at the end of block $\ell$, i.e., $\hat{M}_\ell = \mathcal{D}^{(\ell)}(Y_{\ell 1}^n)$, where $\hat{M}_\ell \in \{1, \ldots, 2^{nR_\ell}\}$. The corresponding (average) probability of error for block code $\ell$ is defined as
\begin{equation*}
  P_e^{(n)}(\ell) = Pr\{\hat{M}_\ell \neq M_\ell\}\:.
\end{equation*}
Rate $R_\ell$ is called achievable if there exists an admissible coding scheme for block code $\ell$ with $\lim_{n\rightarrow\infty} P_e^{(n)}(\ell) = 0$.
\end{definition}
\begin{remark}
The TX encodes each block code according to an individual code-book and information rate. At the beginning of block $\ell$, the TX sends a pilot sequence with power $Q_\ell$ and with some length $n_p \ll n$ to assist the receiver for estimation of $Q_\ell$. From this estimation, the information rate $R_\ell = \frac{1}{2}\log(1+Q_\ell)$ can be derived (See Lemma~\ref{Lem:block}). The estimation error is ignored in this paper. Once $R_\ell$ is calculated by the receiver, the receiver is able to utilize the corresponding code-book designed for rate $R_\ell$ for decoding the remaining $n-n_p \approx n$ information bits.
\end{remark}
\begin{remark}
The decoder decodes message $\hat{M}_\ell = \mathcal{D}^{(\ell)}(Y_{\ell 1}^n)$ after receiving all symbols of the corresponding block $\ell$ without waiting for arrival of future blocks. This is an important difference between this model and the first problem (capacity of the EH-AWGN channel) in \cite{ulukus}: No decision is made after each symbol transmission in \cite{ulukus}. 
\end{remark}
\indent The coding scheme for transmission of the $L-$block code contains a collection of $L$ admissible coding schemes. \\
\indent A power assignment $\{Q_\ell\}_{\ell = 1}^L$ allocated to the block codes is called a power management policy. The \emph{performance} of a policy is measured by (average) throughput
\begin{equation}
  \overline{R_L} = \frac{1}{L}\sum_{\ell = 1}^{L} R_{\ell} \: \label{throughput}
\end{equation}
for a horizon of $L$ data frames. Also in asymptotic case $L\rightarrow \infty$, the asymptotic throughput is defined as
\begin{equation}
  \overline{R_\infty} = \sup_{L\rightarrow \infty} \frac{1}{L}\sum_{\ell = 1}^{L} R_{\ell}\:. \label{throughput-infinit}
\end{equation}

\section{Main Results}   \label{Sec:Results}
Assume that a block code $\ell$ is to be sent from the TX. The following lemma establishes the maximum achievable rate $R_\ell$ in this block code based on the total available energy.

\begin{lem} \label{Lem:block}
Let $B_\ell$ and $E_\ell$ are given and fixed at the beginning of block code $\ell$. Then, any rate
\begin{equation}
  R_\ell = \frac{1}{2}\log(1+ Q_\ell)\: \label{R_l-Lem}
\end{equation}
is an achievable rate for block code $\ell$, where
\begin{equation*}
  Q_\ell \leq E_\ell + \frac{B_\ell}{n}
\end{equation*}
\end{lem}
\begin{proof}
The proof is follows from \cite{ulukus} by some extensions and modifications.  Specific details can be found in~\cite{zib-IWCIT}.
\end{proof}
\begin{remark}
Note that $B_{\ell}$ can grow to infinity with order $n$. This situation happens when a portion of harvested energy rate in previous blocks does not consumed up. So, $\lim_{n\rightarrow \infty} \frac{B_\ell}{n} \neq 0$ in general except $\ell = 1$.
\end{remark}

\indent The main contribution of this paper is the lower bounds on $\overline{R_L}$. However, we first express the upper bound (UB) on $\overline{R_L}$ based on work~\cite{ulukus} to assess the performance of the lower bounds.
\begin{prop}\label{Prop:Upper-Bound}
An upper-bound on $\overline{R_L}$ is given by
\begin{equation}
  \overline{R_L} \leq \frac{1}{2}\log(1+ \frac{1}{L}\sum_{\ell = 1}^L E_{\ell})\: . \label{UB}
\end{equation}
\end{prop}
\begin{proof}
First, assume that the TX has access to the harvested energy sequence $E_{\ell}$ non-causally. Then, \eqref{UB} is derived in \cite[(20)]{ulukus} based on Jensen's inequality \cite{cover}. Second, knowing the sequence $E_{\ell}$ non-causally provides an advantage to the TX generally with respect to the case causal (online) knowledge. Hence, any online policy can achieve a throughput which can not exceed the upper bound on $\overline{R_L}$ of a similar non-causal case.
\end{proof}
\begin{remark}
This UB is generally loose for both the online model and the offline model, because no EC constraint~\eqref{power-const} is taken into account to derive the UB except the power constraint~\eqref{power-const} on the whole block code.
\end{remark}

In this paper, we propose three novel power assignment policies (lower bounds on $\overline{R_L}$) in the following three subsections.

\subsection{Save-And-Transmit (SAT) Across Blocks}
This policy consists of two phases: Save phase And Transmission (SAT) phase. This method is an extension of the SAT across symbols~\cite{ulukus} to address the constraints of the model of this paper.
In the first phase, the harvested energy is saved during $\hbar(L)$ block codes, where $\hbar(L)$ is a function of $L$, and no transmission occurs. In the second phase, transmission takes place with constant power $\overline{\mathcal{E}}$. $\hbar(L)$ is selected such that the battery accumulates sufficient energy $n(L - \hbar(L))\overline{\mathcal{E}}$ during the Save phase such that the TX is able to transmit $L - \hbar(L)$ block codes with constant energy $n\overline{\mathcal{E}}$ each.
In fact, the rates are assigned to the block codes as follows:
\begin{eqnarray}\label{SAT-Rl}
 R_\ell = \left\{
                        \begin{array}{ll}
                            0, & \ell \leq \hbar(L); \\
                            \frac{1}{2}\log(1+P), & \hbar(L) < \ell \leq L.
                        \end{array}
            \right.
\end{eqnarray}
where $P= \overline{\mathcal{E}}-\epsilon$ for any $\epsilon>0$. According to \eqref{SAT-Rl}, the throughput \eqref{throughput} is given by
\begin{equation}
  \overline{R_L} = \frac{L-\hbar(L)}{2L}\log(1+P)\:. \label{SAT-throughput}
\end{equation}
\begin{lem} \label{Lem:SAT}
 Assume that the order of $\hbar(L)$ is smaller than $L$ and $\lim_{L\rightarrow \infty} \hbar(L) = \infty$. Then, the SAT across blocks policy satisfies the power constraint~\eqref{power-const} with high probability (close to one) if $P < \bar{\mathcal{E}}$. Indeed, the SAT across block policy optimally achieves
\begin{equation*}
  \overline{R_\infty} = \frac{1}{2}\log(1 + P)\: .
\end{equation*}
\end{lem}
\begin{proof}
The sketch of the proof is as follows. In the first $\hbar(L)$ blocks, no transmission occurs and the battery collects energy
\begin{eqnarray}
  B_{\hbar(L)} &=& \sum_{k=1}^{\hbar(L)} n E_k \nonumber \\
               &\geq&  n \hbar(L) \overline{\mathcal{E}} - \delta_1 \label{aaa}
\end{eqnarray}
provided $\lim_{L\rightarrow \infty} \hbar(L) = \infty$ due to Strong Law of Large Numbers (SLLN)~\cite{cover}. Now, assume that the battery uses only the stored energy in the battery to transmit at least blocks $\hbar(L)+1$ to $2\hbar(L) - \frac{\delta_1}{n\overline{\mathcal{E}}}$ with power $\overline{\mathcal{E}}-\epsilon$ each. During this period, the stored energy of the first $\hbar(L)$ blocks are completely consumed, but the battery collects new energy $n \hbar(L) \overline{\mathcal{E}} - \delta_2$ similar to \eqref{aaa}. Again, this energy can afford the transmission up to at least block $3\hbar(L) - \frac{\delta_2}{n\overline{\mathcal{E}}}$ with power $\overline{\mathcal{E}}-\epsilon$ each. This iteration can happen as long as all blocks after block $\hbar(L)$  are sent with no energy outage. Hence, the transmission phase requirements are met. The lemma is concluded from \eqref{SAT-throughput} because only the first $\hbar(L)$ blocks are not sent.
\end{proof}

\subsection{Best-Effort-Transmit (BET) Across Blocks}
In this policy, the TX does best effort to transmit a block code with a given constant power $P < \overline{\mathcal{E}}$. If sufficient energy is available for the whole block transmission, that block will be sent. Otherwise, the TX does not send the block and stores energy for future.
This policy is an extension of the BET across symbol~\cite{ulukus} to address the constraints of the model of this paper. Specifically, a Gaussian code-book $\mathcal{C}$ is generated according to $\mathcal{N}(0, P)$. Block code $\ell$ is sent only if $B_\ell + n E_\ell \geq n P$. If block code $\ell$ is sent,
\begin{eqnarray*}
      Q_\ell &=& P \\
      B_{\ell + 1} &=& B_\ell + n E_\ell - nP
\end{eqnarray*}
according to Lemma~\ref{Lem:block}. Otherwise, the block code does not send and $B_{\ell + 1} = B_\ell + n E_\ell$. The following lemma approves that almost all block codes are sent.

\begin{lem}   \label{Lem:BET}
In the BET across blocks policy, if $P < \bar{\mathcal{E}}$ and $L\rightarrow \infty$, the scheme optimally achieves
\begin{equation*}
  \overline{R_\infty} = \frac{1}{2}\log(1+ P) \: .
\end{equation*}
\end{lem}
\begin{proof}
First, assume there exists a $k_0$ beyond which all block codes are sent with power $P$ such that $k_0$ is a function of $L$ with some increasing order less than $O(L)$. Ignoring the first $k_0$ blocks in~\eqref{throughput-infinit}, we obtain the following lower bound
\begin{equation}\label{BET:first_assump}
  \overline{R_\infty} \geq \lim_{L\rightarrow \infty} \frac{L - k_0}{2L}\log(1 + P) = \frac{1}{2}\log(1 + P)
\end{equation}
Hence, the contribution of the first $k_0$ blocks in the average throughput is negligible as $L$ grows.
Second, assume there exists some $k_0$ with increasing order $L$ or higher as a function of $L$, such that $L^{th}$ block code is the $k_0^{th}$ block code with transmission power $Q_L = 0$. In other words, $L - k_0$ blocks before $L^{th}$ block are sent by transmission power $P$ and $k_0-1$ of them as well as $L^{th}$ block are not sent. If we deduct the consumed energy from the total stored energy during $L$ blocks, the following lower bound on $B_{L} + nE_L$ is obtained.
\begin{eqnarray}
  B_{L} + nE_L &=& \sum_{\ell=1}^{L-1} nE_\ell - n(L - k_0)P + n E_L \label{BET:proof1}\\
   &>& nL(\bar{\mathcal{E}} - \delta_L) - n(L - k_0)\bar{\mathcal{E}} \label{BET:proof2} \\
   &=& -nL \delta_L + nk_0\bar{\mathcal{E}} \nonumber\\
   &\geq& nP            \label{BET:proof3}
\end{eqnarray}
where \eqref{BET:proof1} follows from the fact that $L-k_0$ blocks are sent by power $P$; \eqref{BET:proof2} follows from $P< \bar{\mathcal{E}}$ and from Strong Law of Large Numbers (SLLN)~\cite{cover}; \eqref{BET:proof3} holds for \emph{any} $k_0 \geq \frac{L \delta_L + P}{\bar{\mathcal{E}}}$. The order of $L\delta_L$ is strictly less than $L$~\cite{cover}, and thus we conclude that \eqref{BET:proof3} is met for any $k_0$ with order $L$ or higher as $L\rightarrow \infty$. According to \eqref{BET:proof3}, the battery has enough energy to afford transmission in the $L^{th}$ block code. This result violates the initial assumption $O(k_0)\geq L$. The lemma is concluded from the first assumption.
\end{proof}
\begin{remark}
The SAT across blocks and the BET across blocks are extended versions of the SAT across symbols and the BET across symbols which have been developed by \cite{ulukus} for the capacity problem. In this paper, these extended versions are employed for the standard online throughput maximization problem though. Indeed, the corresponding proofs are completely novel, because extending the proofs of \cite{ulukus} to this work is not trivial because $\{Q_\ell\}_{\ell=1}^L$ is not an i.i.d. sequence.
\end{remark}
\subsection{Adaptive Power Allocation (APA) Across Blocks}
In this policy, the TX adaptively allocate energy arrivals to the block codes. Set a constant power $P < \overline{\mathcal{E}}$. For any block code, if the total available energy (including the stored energy in the battery and the harvested energy during the block transmission) is sufficient to afford the block transmission with power $P$, then energy $nP$ is allocated to that block code and the extra energy remains in the battery for future usage. Otherwise, the total available energy is allocated to that block. Specifically, if $B_\ell + n E_\ell \geq n P$, then
\begin{eqnarray*}
      Q_\ell &=& P \\
      B_{\ell + 1} &=& B_\ell + n E_\ell - nP \:,
\end{eqnarray*}
and the block code is called a perfect block code. If $B_\ell + n E_\ell < n P$, then
\begin{eqnarray}
      Q_\ell &=& \frac{B_\ell}{n} + E_\ell \\
      B_{\ell + 1} &=& 0 \:,   \label{APA:Battery2}
\end{eqnarray}
and the block code is called imperfect.
\begin{lem}   \label{Lem:APA}
In the APA across blocks policy, if $P < \bar{\mathcal{E}}$ and $L\rightarrow \infty$, the scheme optimally achieves
\begin{equation*}
  \overline{R_\infty} = \frac{1}{2}\log(1 + P)\: .
\end{equation*}
\end{lem}
\begin{proof}
First, assume that there exists a $k_0$ beyond which all block codes are sent with power $P$ with property $O(k_0)<O(L)$. Similar to proof of Lemma~\ref{Lem:BET}, \eqref{BET:first_assump} is derived.
Second, assume there exists some $k_0$ such that $L^{th}$ block code is the $k_0^{th}$ block code with transmission power $Q_L < P$, where the order of $k_0$ is $L$ or higher.
In other words, $L - k_0$ blocks before $L^{th}$ block are sent by perfect transmission power $P$ and $k_0-1$ of them as well as $L^{th}$ block are sent by some imperfect power strictly smaller than $P$. According to~\eqref{APA:Battery2}, the battery is required to be empty after block code $L^{th}$, i.e. $B_{L+1} = 0$. On the other hand, if we deduct the maximum consumed energy from the total stored energy during $L$ blocks, we can obtain the following lower bound on $B_{L+1}$.
\begin{eqnarray}
  B_{L+1} &\geq& \sum_{\ell=1}^L nE_\ell - n(L - k_0)P - nk_0 (P-\epsilon)  \label{APA:proof1}\\
   &\geq& nL(\bar{\mathcal{E}} - \delta_L) - nLP + nk_0 \epsilon \label{APA:proof2} \\
   &>& -nL \delta_L + nk_0 \epsilon  \label{APA:proof3}\\
   &\geq& 0    \label{APA:proof4}
\end{eqnarray}
where \eqref{APA:proof1} follows by assuming that all the $k_0$ imperfect blocks are sent by maximum possible power $P - \epsilon$, where $\epsilon>0$ is an arbitrary \emph{fixed} constant independent of $k_0$ and $L$; \eqref{APA:proof2} follows from Strong Law of Large Numbers (SLLN)~\cite{cover}; \eqref{APA:proof3} follows from the fact that $\bar{\mathcal{E}} > P$; \eqref{APA:proof4} holds for any $k_0 \geq \frac{L \delta_L}{\epsilon}$. Order of $L \delta_L$ does not exceed $L$~\cite{cover}, and thus any $k_0$ with order at least $L$ satisfies~\eqref{APA:proof4} as $L\rightarrow \infty$. This result violates~\eqref{APA:Battery2}, and thus such a $k_0$ with $O(k_0)\geq L$ does not exist. The lemma is concluded from the first assumption.
\end{proof}
\begin{theorem}
The SAT over block policy, the BET over block policy, and the APA policy optimally achieves
\begin{equation*}
  \overline{R_\infty} = \frac{1}{2}\log(1+ \overline{\mathcal{E}})
\end{equation*}
for the asymptotic case $L\rightarrow \infty$.
\end{theorem}
\begin{proof}
The proof follows from Lemma~\ref{Lem:SAT}, Lemma~\ref{Lem:BET}, and Lemma~\ref{Lem:APA}, respectively, because the optimal throughput $\overline{R_\infty}$ is achieved if $P\rightarrow \overline{\mathcal{E}}^-$ (from the left) and $L\rightarrow \infty$.
\end{proof}
\section{Numerical Results} \label{Sec:Numerical-Results}
In this section, we present the numerical results of this research. We have compared six methods in the following figures:  A naive method (based on no power management) which will be explained in the sequel, the three proposed policies of this paper, the Optimal Power Management (OPM) for the corresponding offline model given in~\cite{ulukus, modiano, yener},  and the upper-bound (UB) given in Prop.~\eqref{Prop:Upper-Bound}.\\
\indent Following~\cite{ulukus}, we have assumed an exponential distribution for the recharge rate sequence $\{E_\ell\}_{\ell=1}^L$. In Fig.~\ref{Fig2-Exponential1}, the throughput is sketched for each method versus different mean values $\overline{\mathcal{E}}$. The total number of block codes in Fig.~\ref{Fig2-Exponential1}  is fixed to $L = 500$. In Fig.~\ref{Fig3-Exponential1}, the average throughput is sketched as a function of $L$ in a semi-log plot when the mean of the recharge sequence is fixed to $\overline{\mathcal{E}} = 10$. All the curves of these figures are generated by averaging over $1000$ runs.\\
\indent A naive policy assigns the harvested energy $E_\ell$ to block code $\ell$ entirely and leaves no energy at the end of the block code in the battery. Specifically, the native method assigns
\begin{eqnarray*}
    Q_\ell &=&  E_\ell\\
    R_\ell &=& \frac{1}{2}\log(1+ E_\ell)
\end{eqnarray*}
which leads to $\overline{R_L} = \frac{1}{2}\sum_{\ell = 1}^L \log(1+ E_\ell)$. From Fig.~\ref{Fig3-Exponential1},  the corresponding $\overline{R_L}$ is a constant function of $L$, which is the statistical average of $R_\ell$ because $\{R_\ell\}_{\ell = 1}^L$ is an i.i.d. sequence is this case. When $L$ is not large enough, $L = 50$ for example, the naive method outperforms the SAT across blocks according to Fig.~\ref{Fig3-Exponential1}. Because the save phase in the SAT method takes many time frames to charge up the battery for the transmission phase, and thus a considerable portion of the block codes $(\frac{\hbar(L)}{L})$ remains silent. However, for large enough block codes, $L = 500$ for example, the SAT across blocks outperforms the naive method according to Fig.~\ref{Fig2-Exponential1}, because the save phase contains negligible portion of the whole block codes.\\
\indent Due to Fig.~\ref{Fig2-Exponential1}, and Fig.~\ref{Fig3-Exponential1}, the BET over blocks policy outperforms the SAT over blocks policy, and the APA policy outperforms the BET over blocks policy in general.
From Fig.~\ref{Fig3-Exponential1}, the three proposed online policies and the OPM offline asymptotically converges to the UB, and thus all achieves the optimal  throughput.
\begin{figure}[h]
  \centering
  \includegraphics[width=9cm]{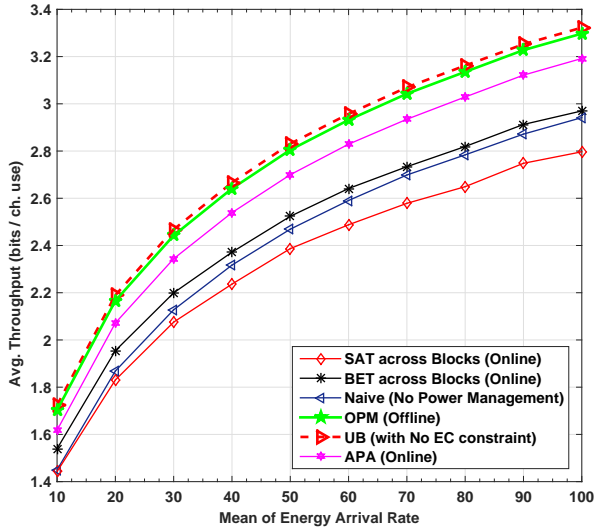}
  \caption{Comparison the performance of different methods as a function of the mean of $P_E$ (exponential distribution) for $L = 50$.}\label{Fig1-Exponential1}
\end{figure}
\begin{figure}[h]
  \centering
  \includegraphics[width=9cm]{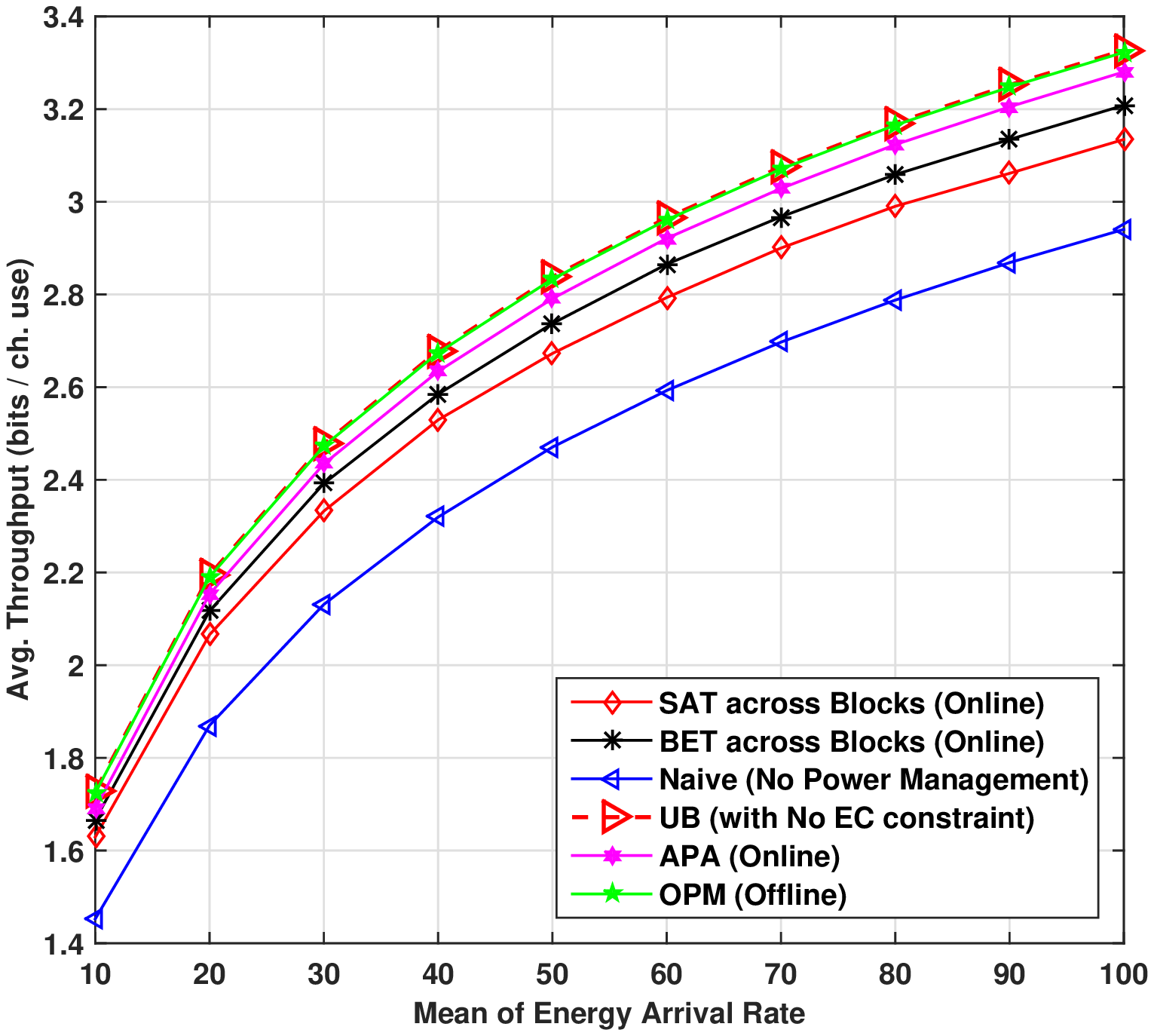}
  \caption{Comparison of the performance of different methods as a function of the mean of $P_E$ (exponential distribution) for $L = 500$.}\label{Fig2-Exponential1}
\end{figure}

\begin{figure}[h]
  \centering
  \includegraphics[width=9cm]{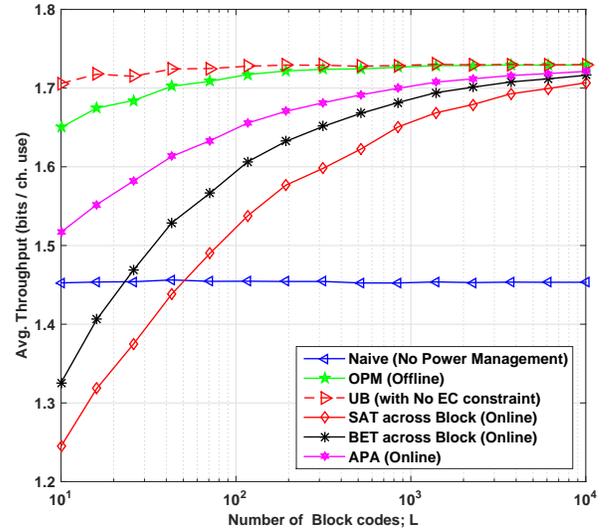}
  \caption{Comparison of the performance of different methods as a function of the number of blocks $(L)$ when the mean of $P_E$ (exponential distribution) is fixed to $10$.}\label{Fig3-Exponential1}
\end{figure}

\section{Conclusions}  \label{Sec:Conclusion}
In this paper, we have developed three novel schemes for the EH problem over an AWGN channel with slow-varying harvested energy. The objective is to acquire optimal power assignment (policy) across blocks to achieve the maximum throughput. Three novel online policies are developed in this paper. The schemes are simple in complexity such that the allocation power of each block code can be uniquely determined from the current energy arrival rate and the battery state. The schemes are asymptotically optimal as the number of block codes grows. The schemes can be exploited as simple efficient offline policies as well.
\bibliographystyle{IEEEtran}
\bibliography{EH}
\end{document}